\documentclass{llncs}

\usepackage{amsmath}
\usepackage{amssymb}
\usepackage[dvipdfmx]{graphicx}

\newcommand{\YES}{\mathsf{yes}}
\newcommand{\NO}{\mathsf{no}}

\newcommand{\anyds}{D}
\newcommand{\canods}{C}
\newcommand{\vercano}[1]{w_{#1}}

\newcommand{\parta}{V_1}
\newcommand{\partb}{V_2}
\newcommand{\partc}{V_3}

\newcommand{\calI}{\mathcal{I}}

\newcommand{\intervala}[1]{l(#1)}
\newcommand{\intervalb}[1]{r(#1)}

\newcommand{\sevstep}[1]{\overset{#1}{\leftrightsquigarrow}}

\newcommand{\cocompg}[1]{G_{#1}}
\newcommand{\cocompnum}{p}
\newcommand{\coga}{G_a}
\newcommand{\cogb}{G_b}
\newcommand{\cographa}{w_a}
\newcommand{\cographb}{w_b}
\newcommand{\cocanods}{C}

\newenvironment{listing}[1]{%
    \begin{list}{*}{%
    \settowidth{\labelwidth}{#1}%
    \setlength{\leftmargin}{\labelwidth}%
    \advance \leftmargin by 12pt
    \setlength{\itemsep}{0pt}%
    \setlength{\parsep}{0pt}%
    \setlength{\topsep}{0pt}%
    \setlength{\parskip}{0pt}%
    }%
    }{%
    \end{list}}

\begin{document}
	\title{The complexity of\\ dominating set reconfiguration}

\author{%
	Arash Haddadan\inst{1} \and
	Takehiro Ito\inst{2} \and
	Amer E. Mouawad\inst{1} \and\\
	Naomi Nishimura\inst{1} \and
	Hirotaka Ono\inst{3} \and
	Akira Suzuki\inst{2} \and
	Youcef Tebbal\inst{1}
}

\institute{%
	University of Waterloo,\\
	200 University Ave. West, Waterloo, Ontario N2L 3G1, Canada.\\
	\email{\{ahaddada, aabdomou, nishi, ytebbal\}@uwaterloo.ca}
\and
	Graduate School of Information Sciences, Tohoku University, \\
	Aoba-yama 6-6-05, Sendai, 980-8579, Japan.\\
	\email{\{takehiro, a.suzuki\}@ecei.tohoku.ac.jp}
\and
	Faculty of Economics, Kyushu University, \\
	Hakozaki 6-19-1, Higashi-ku, Fukuoka, 812-8581, Japan.\\
	\email{hirotaka@econ.kyushu-u.ac.jp}
}

\maketitle

\begin{abstract}
Suppose that we are given two dominating sets $D_s$ and $D_t$ of a graph $G$
whose cardinalities are at most a given threshold $k$.
Then, we are asked whether there exists a sequence of dominating sets of $G$ between
$D_s$ and $D_t$ such that each dominating set in the sequence is of cardinality
at most $k$ and can be obtained from the previous one by either adding or deleting exactly one vertex.
This problem is known to be PSPACE-complete in general.
In this paper, we study the complexity of this decision problem from the viewpoint of graph classes.
We first prove that the problem remains PSPACE-complete even for planar graphs, bounded bandwidth graphs, split graphs, and bipartite graphs.
We then give a general scheme to construct linear-time algorithms and show
that the problem can be solved in linear time for cographs, trees, and interval graphs.
Furthermore, for these tractable cases, we can obtain a desired sequence such
that the number of additions and deletions is bounded by $O(n)$, where $n$ is the number of vertices in the input graph.
\end{abstract}

\section{Introduction}

Consider the art gallery problem modeled on graphs:
Each vertex corresponds to a room which has a monitoring camera
and each edge represents the adjacency of two rooms.
Assume that each camera in a room can monitor the room itself and its adjacent rooms.
Then, we wish to find a subset of cameras that can monitor all rooms;
the corresponding vertex subset $D$ of the graph $G$ is called a {\em dominating set},
that is, every vertex in $G$ is either in $D$ or adjacent to a vertex in $D$.
For example, \figurename~\ref{fig:example} shows six different dominating sets of the same graph.
Given a graph $G$ and a positive integer $k$, the problem of determining
whether $G$ has a dominating set of cardinality at most
$k$ is a classical NP-complete problem~\cite{GJ79}.

\subsection{Our problem}

However, the art gallery problem could be considered in more ``dynamic'' situations:
In order to maintain the cameras, we sometimes need to change the current dominating set into another one.
This transformation needs to be done by switching the cameras individually and we certainly need to
keep monitoring all rooms, even during the transformation.

In this paper, we thus study the following problem:
Suppose that we are given two dominating sets of a graph $G$ whose cardinalities are at most
a given threshold $k > 0$ (e.g., the leftmost and rightmost ones in \figurename~\ref{fig:example}, where $k = 4$), and
we are asked whether we can transform one into the other via dominating sets of $G$ such that each
intermediate dominating set is of cardinality at most $k$ and can be obtained from the
previous one by either adding or deleting a single vertex.
We call this decision problem the {\sc dominating set reconfiguration (DSR)} problem.
For the particular instance of \figurename~\ref{fig:example}, the answer is $\YES$ as illustrated in \figurename~\ref{fig:example}.

\begin{figure}[t]
	\centering
	\includegraphics[width=0.9\linewidth]{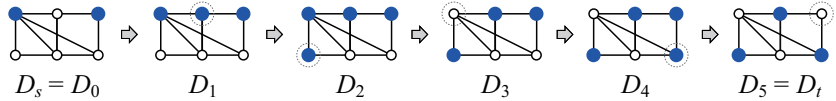}
	\vspace{-1em}
	\caption{A sequence $\langle D_0, D_1, \ldots, D_5 \rangle$ of dominating sets in the
    same graph, where $k = 4$ and the vertices in dominating sets are depicted by large (blue) circles.}
	\vspace{-1em}
	\label{fig:example}
\end{figure}
	
\subsection{Known and related results}

Recently, similar problems have been extensively studied under the
reconfiguration framework~\cite{IDHPSUU}, which arises when we
wish to find a step-by-step transformation between two feasible solutions of
a combinatorial problem such that all intermediate solutions are also feasible.
The reconfiguration framework has been applied to several well-studied problems, including
{\sc satisfiability}~\cite{Kolaitis},
{\sc independent set}~\cite{HearnDemaine2005,IDHPSUU,KaminskiMM12,MNRSS13,Wro14},
{\sc vertex cover}~\cite{IDHPSUU,INZ14,MNR14,MNRSS13},
{\sc clique}, {\sc matching}~\cite{IDHPSUU},
{\sc vertex-coloring}~\cite{BC09},
and so on. (See also a survey~\cite{van13}.)

Mouawad et al.~\cite{MNRSS13} proved that {\sc dominating set reconfiguration} is $W[2]$-hard
when parameterized by $k+\ell$, where $k$ is the cardinality threshold
of dominating sets and $\ell$ is the length of a sequence of dominating sets.
	
Haas and Seyffarth~\cite{HS14} gave sufficient conditions for the
cardinality threshold $k$ for which any two dominating sets can be transformed into one another.
They proved that the answer to {\sc dominating set reconfiguration} is $\YES$ for
a graph $G$ with $n$ vertices if $k = n - 1$ and $G$ has a matching of cardinality at least two;
they also gave a better sufficient condition when restricted to bipartite or chordal graphs.
Recently, Suzuki et al.~\cite{SMN14} improved the former condition and showed
that the answer is $\YES$ if $k = n - \mu$ and $G$ has a matching of
cardinality at least $\mu+1$, for any nonnegative integer $\mu$.

\subsection{Our contribution}
	
To the best of our knowledge, no algorithmic results are known
for the {\sc dominating set reconfiguration} problem and it is therefore desirable
to obtain a better understanding of what separates ``hard'' from ``easy'' instances.
To that end, we study the problem from the viewpoint of
graph classes and paint an interesting picture of the boundary
between intractability and polynomial-time solvability.
(See also \figurename~\ref{fig:results}.)

We first prove that the problem is PSPACE-complete even on
planar graphs, bounded bandwidth graphs, split graphs, and bipartite graphs.
Our reductions for PSPACE-hardness follow from the classical reductions
for proving the NP-hardness of {\sc dominating set}.
However, the reductions should be constructed carefully so that they preserve
not only the existence of dominating sets but also the reconfigurability.
	
We then give a general scheme to construct linear-time algorithms for the problem.
As examples of its application, we demonstrate that the problem can
be solved in linear time on cographs (also known as $P_4$-free graphs), trees, and interval graphs.
Furthermore, for these tractable cases, we can obtain a desired sequence such
that the number of additions and deletions (i.e., the length of a reconfiguration sequence) can
be bounded by $O(n)$, where $n$ is the number of vertices in the input graph.
	
Proofs of lemmas and theorems marked with a star can be found in the appendix.

\begin{figure}[t]
    \centering
	\includegraphics[width=\linewidth]{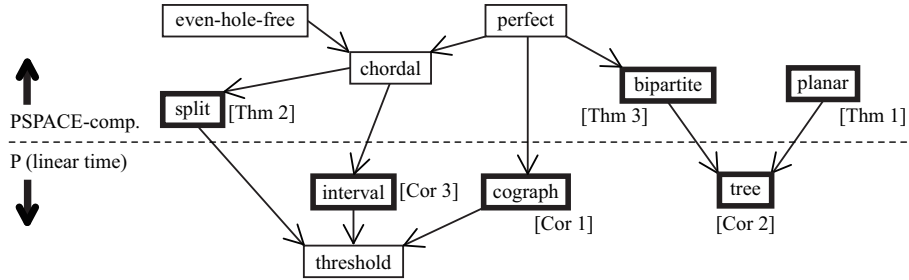}
	\vspace{-1em}
	\caption{Our results, where each arrow represents the inclusion relationship between graph classes:
			$A \to B$ represents that $B$ is properly included in $A$~\cite{BLS99}.
			We also show PSPACE-completeness on graphs of bounded bandwidth (Theorem~\ref{the:hardness1}).}
	\vspace{-1em}
    \label{fig:results}
\end{figure}

\section{Preliminaries}

	
\subsubsection{Graph notation and dominating set.}

We assume that each input graph $G$ is a simple undirected graph with vertex set $V(G)$
and edge set $E(G)$, where $|V(G)| = n$ and $|E(G)| = m$.
For a set $S \subseteq V(G)$ of vertices, the subgraph of $G$ {\em induced} by $S$ is
denoted by $G[S]$, where $G[S]$ has vertex set $S$ and edge set $\{uv \in E(G) \mid u, v \in S\}$.

For a vertex $v$ in a graph $G$, we let $N_G(v) = \{u \in V(G) \mid vu \in E(G)\}$ and $N_G[v] = N_G(v) \cup \{v\}$.
For a set $S \subseteq V(G)$ of vertices, we define $N_G[S] = \bigcup_{v \in S} N_G[v]$ and $N_G(S) = N_G[S] \setminus S$.
We sometimes drop the subscript $G$ if it is clear from the context.

For a graph $G$, a set $D \subseteq V(G)$ is a {\em dominating set} of $G$ if $N_G[D] = V(G)$.
Note that $V(G)$ always forms a dominating set of $G$.
For a vertex $u \in V(G)$ and a dominating set $D$ of $G$, we say that $u$ is {\em dominated} by $v \in D$ if $u \notin D$ and $u \in N_G(v)$.
A vertex $w$ in a dominating set $D$ is {\em deletable} if $D \setminus \{w\}$ is also a dominating set of $G$.
A dominating set $D$ of $G$ is {\em minimal} if there is no deletable vertex in $D$.

\subsubsection{Dominating set reconfiguration.}
	
We say that two dominating sets $D$ and $D^\prime$ of the same graph $G$
are {\em adjacent} if there exists a vertex $u \in V(G)$ such that
$D \vartriangle D^\prime = (D \setminus D^\prime) \cup (D^\prime \setminus D) = \{ u\}$,
i.e. $u$ is the only vertex in the {\em symmetric difference} of $D$ and $D^\prime$.
For two dominating sets $D_p$ and $D_q$ of $G$, a sequence $\langle D_0, D_1, \ldots, D_\ell \rangle$ of
dominating sets of $G$ is called a {\em reconfiguration sequence} between $D_p$ and $D_q$ if it has the following properties:
\begin{listing}{aaa}
\item[(a)] $D_0 = D_p$ and $D_\ell = D_q$; and
\item[(b)] $D_{i-1}$ and $D_i$ are adjacent for each $i \in \{1,2,\ldots, \ell\}$.
\end{listing}
Note that any reconfiguration sequence is {\em reversible}, that is, $\langle D_\ell, D_{\ell-1}, \ldots, D_0 \rangle$ is
also a reconfiguration sequence between $D_p$ and $D_q$.
We say a vertex $v \in V(G)$ is {\em touched} in a reconfiguration
sequence $\sigma = \langle D_0, D_1, \ldots, D_\ell \rangle$ if $v$ is either added or deleted at least once in $\sigma$.

For two dominating sets $D_p$ and $D_q$ of a graph $G$ and an integer $k>0$, we
write $D_p \sevstep{k} D_q$ if there exists a reconfiguration sequence $\langle D_0, D_1, \ldots, D_\ell \rangle$
between $D_p$ and $D_q$ in $G$ such that $|D_i| \le k$ holds for every $i \in \{0, 1, \ldots, \ell\}$, for some $\ell \geq 0$.
Note that $k \ge \max \{|D_p|, |D_q|\}$ clearly holds if $D_p \sevstep{k} D_q$.
Then, the {\sc dominating set reconfiguration (DSR)} problem is defined as follows:	
\begin{center}
	\parbox{0.85\hsize}{
	\begin{listing}{{\bf Question:}}
	\item[{\bf Input:}] A graph $G$, two dominating sets $D_s$ and $D_t$ of $G$, and an integer threshold $k \geq \max \{ |D_s|, |D_t| \}$
	\item[{\bf Question:}] Determine whether $D_s \sevstep{k} D_t$ or not.
	\end{listing}}
\end{center}
We denote by a $4$-tuple $(G, D_s, D_t, k)$ an instance of {\sc dominating set reconfiguration}.
Note that {\sc DSR} is a decision problem and hence it does not ask for an actual reconfiguration sequence.
We always denote by $D_s$ and $D_t$ the {\em source} and {\em target} dominating sets of $G$, respectively.

\section{PSPACE-completeness}\label{sec:hardness}

In this section, we prove that {\sc dominating set reconfiguration} remains PSPACE-complete even for restricted classes of graphs;
some of these classes show nice contrasts to our algorithmic results in Section~\ref{sec:algo}.
(See also \figurename~\ref{fig:results}.)

\begin{theorem}\label{the:hardness1}
{\sc DSR} is PSPACE-complete on planar graphs of maximum degree six and on graphs of bounded bandwidth.
\end{theorem}

\begin{proof}
One can observe that the problem is in PSPACE~\cite[Theorem~1]{IDHPSUU}.
We thus show that it is PSPACE-hard for those graph classes by a polynomial-time
reduction from {\sc vertex cover reconfiguration}~\cite{IDHPSUU,INZ14,MNR14}.
In {\sc vertex cover reconfiguration}, we are given two vertex covers $C_s$ and $C_t$ of a graph
$G^\prime$ such that $|C_s| \le k$ and $|C_t| \le k$, for some integer $k$, and asked
whether there exists a reconfiguration sequence of vertex covers
$C_0, C_1, \ldots, C_\ell$ of $G$ such that $C_0 = C_s$, $C_\ell = C_t$, $|C_i| \le k$,
and $|C_{i-1} \vartriangle C_i| = 1$ for each $i \in \{1, 2, \ldots, \ell\}$.
	
Our reduction follows from the classical reduction from {\sc vertex cover} to {\sc dominating set}~\cite{GJ79}.
Specifically, for every edge $uw$ in $E(G^\prime)$, we add a new vertex $v_{uw}$
and join it with each of $u$ and $w$ by two new edges $u v_{uw}$ and $v_{uw}w$;
let $G$ be the resulting graph.
Then, let $(G, D_s = C_s, D_t = C_t, k)$ be the corresponding instance of {\sc dominating set reconfiguration}.
Clearly, this instance can be constructed in polynomial time.

We now prove that $D_s \sevstep{k} D_t$ holds if and only if there
is a reconfiguration sequence of vertex covers in $G^\prime$ between $C_s$ and $C_t$.
However, the if direction is trivial, because any vertex cover of
$G^\prime$ forms a dominating set of $G$ and both problems employ the same
reconfiguration rule (i.e., the symmetric difference is of size one).
Therefore, suppose that $D_s \sevstep{k} D_t$ holds, and hence there exists
a reconfiguration sequence of dominating sets in $G$ between $D_s$ and $D_t$.
Recall that neither $D_s$ nor $D_t$ contain a newly added vertex in $V(G) \setminus V(G^\prime)$.
Thus, if a vertex $v_{uw}$ in $V(G) \setminus V(G^\prime)$ is touched, then $v_{uw}$ must be added first.
By the construction of $G$, both $N_G[v_{uw}] \subseteq N_G[u]$ and $N_G[v_{uw}] \subseteq N_G[w]$ hold.
Therefore, we can replace the addition of $v_{uw}$ by that of either
$u$ or $w$ and obtain a (possibly shorter) reconfiguration sequence
of dominating sets in $G$ between $D_s$ and $D_t$ which touches vertices only in $G^\prime$.
Then, it is a reconfiguration sequence of vertex covers in $G^\prime$
between $C_s$ and $C_t$, as needed.
	
{\sc Vertex cover reconfiguration} is known to be PSPACE-complete
on planar graphs of maximum degree three~\cite{INZ14,MNR14} and on graphs of bounded bandwidth~\cite{Wro14}.
Thus, the reduction above implies PSPACE-hardness on planar
graphs of maximum degree six and on graphs of bounded bandwidth;
note that, since the number of edges in $G$ is only the
triple of that in $G^\prime$, the bandwidth increases only by a constant multiplicative factor.
\qed
\end{proof}

We note that both pathwidth and treewidth
of a graph $G$ are bounded by the bandwidth of $G$.
Thus, Theorem~\ref{the:hardness1} yields that
{\sc dominating set reconfiguration} is PSPACE-complete on
graphs of bounded pathwidth and treewidth.

Adapting known techniques from NP-hardness proofs for the {\sc dominating set} problem~\cite{Ber84}, 
we also show PSPACE-completeness of {\sc dominating set reconfiguration}
on split graphs and on bipartite graphs; a graph is {\em split} if its vertex set can be
partitioned into a clique and an independent set~\cite{BLS99}.

\begin{theorem}[*]\label{the:split}
{\sc DSR} is PSPACE-complete on split graphs.
\end{theorem}

\begin{theorem}[*]\label{the:bipartite}
{\sc DSR} is PSPACE-complete on bipartite graphs.
\end{theorem}

\section{General scheme for linear-time algorithms} \label{sec:algo}

In this section, we show that {\sc dominating set reconfiguration} is
solvable in linear time on cographs, trees, and interval graphs.
Interestingly, these results can be obtained by the application
of the same strategy; we first describe the general scheme in Section~\ref{dsr:genestra}.
We then show in Sections~\ref{dsr:cograph}--\ref{dsr:alinterval} that the
problem can be solved in linear time on those graph classes.

\subsection{General scheme}\label{dsr:genestra}
The general idea is to introduce the concept of a ``canonical'' dominating set for a graph $G$.
We say that a minimum dominating set $\canods$ of $G$ is {\em canonical}
if $\anyds \sevstep{k} \canods$ holds for every dominating set $\anyds$ of $G$ and $k = |\anyds|+1$.
Then, we have the following theorem.

\begin{theorem} \label{the:canonical}
If a graph $G$ has a canonical dominating set, then {\sc dominating set reconfiguration} can be solved in linear time on $G$ .
\end{theorem}

We note that proving the existence of a canonical dominating set is sufficient for solving the decision problem.
Therefore, we do not need to find an actual canonical dominating set in linear time.
In Sections~\ref{dsr:cograph}--\ref{dsr:alinterval}, we will show that cographs, trees, and interval graphs admit
canonical dominating sets, and hence the problem can be solved in linear time on those graph classes.
Note that, however, Theorem~\ref{the:canonical} can be applied to any graph which has a canonical dominating set.
In the remainder of this subsection, we prove Theorem~\ref{the:canonical} starting 
with the following lemma.	

\begin{lemma} \label{lem:plusone}
Suppose that a graph $G$ has a canonical dominating set.
Then, an instance $(G, D_s, D_t, k)$ of {\sc dominating set reconfiguration} is a $\YES$-instance if $k \ge \max \{|D_s|, |D_t|\}+1$.
\end{lemma}

\begin{proof}
Let $\canods$ be a canonical dominating set of $G$.
Then, $D_s \sevstep{k^\prime} \canods$ holds for $k^\prime = |D_s|+1$.
Suppose that $k \ge \max \{|D_s|, |D_t|\}+1$.
Since $k \ge |D_s| + 1 = k^\prime$, we clearly have $D_s \sevstep{k} \canods$.
Similarly, we have $D_t \sevstep{k} \canods$.
Since any reconfiguration sequence is reversible, we have $D_s \sevstep{k} \canods \sevstep{k} D_t$, as needed.
\qed
\end{proof}
	
Lemma~\ref{lem:plusone} implies that if a graph $G$ has a canonical
dominating set $\canods$, then it suffices to consider the case where $k = \max \{ |D_s|, |D_t|\}$.
Note that there exist $\NO$-instances of {\sc dominating set reconfiguration} in
such a case but we show that they can be easily identified in linear time, as implied by the following lemma.

\begin{lemma} \label{lem:pluszero}
Let $(G, D_s, D_t, k)$ be an instance of {\sc dominating set reconfiguration},
where $G$ is a graph admitting a canonical dominating set and $k = \max \{|D_s|, |D_t|\}$.
Then, $(G, D_s, D_t, k)$ is a $\YES$-instance if and only if
$D_i$ is not minimal for every $i \in \{s,t\}$ such that $|D_i| = k$.
\end{lemma}

Lemma~\ref{lem:pluszero} can be immediately obtained from the following lemma.

\begin{lemma} \label{lem:minimal}
Suppose that a graph $G$ has a canonical dominating set $\canods$.
Let $\anyds$ be an arbitrary dominating set of $G$ and let $k=|\anyds|$.
Then, $\anyds \sevstep{k} \canods$ holds if and only if $\anyds$ is not a minimal dominating set.
\end{lemma}

\begin{proof}
{\em Necessity.}
Suppose that $\anyds$ is not minimal.
Then, $\anyds$ contains at least one vertex $x$ which is deletable
from $\anyds$, that is, $\anyds \setminus \{x\}$ forms a dominating set of $G$.
Since $k = |\anyds| = |\anyds \setminus \{x\}| + 1$, we have $\anyds \setminus \{x\} \sevstep{k} \canods$.
Therefore, $\anyds \sevstep{k} \anyds \setminus \{x\} \sevstep{k} \canods$ holds.

\noindent	
{\em Sufficiency.}
We prove the contrapositive.
Suppose that $\anyds$ is minimal.
Then, no vertex in $\anyds$ is deletable and hence any dominating
set $\anyds^\prime$ which is adjacent to $\anyds$ must be obtained by adding a vertex to $\anyds$.
Therefore, $|\anyds^\prime| = k+1$ for any dominating
set $\anyds^\prime$ which is adjacent to $\anyds$. Hence, $\anyds \sevstep{k} \canods$ does not hold.
\qed
\end{proof}
	
We note again that Lemmas~\ref{lem:plusone} and \ref{lem:pluszero} imply that
an actual canonical dominating set is not required to solve the problem.
Furthermore, it can be easily determined in linear time whether a dominating set of a graph $G$ is minimal or not.
Thus, Theorem~\ref{the:canonical} follows from Lemmas~\ref{lem:plusone} and \ref{lem:pluszero}.
\medskip

Before constructing canonical dominating sets in
Sections~\ref{dsr:cograph}--\ref{dsr:alinterval}, we give the following lemma showing that
it suffices to construct a canonical dominating set for a connected graph.

\begin{lemma}[*]\label{lem:conncted}
Let $G$ be a graph consisting of $\cocompnum$ connected components $\cocompg{1}, \cocompg{2}, \ldots, \allowbreak \cocompg{\cocompnum}$.
For each $i \in \{1,2,\ldots, \cocompnum\}$, suppose that $\cocanods_i$ is a canonical dominating set for $\cocompg{i}$.
Then, $\cocanods = \cocanods_1 \cup \cocanods_2 \cup \cdots \cup \cocanods_{\cocompnum}$ is a canonical dominating set for $G$.
\end{lemma}

\subsection{Cographs} \label{dsr:cograph}
	
We first define the class of cographs (also known as $P_4$-free graphs)~\cite{BLS99}.
For two graphs $G_1$ and $G_2$, their {\em union $G_1 \cup G_2$} is the graph such
that $V(G_1 \cup G_2) = V(G_1) \cup V(G_2)$ and $E(G_1 \cup G_2) = E(G_1) \cup E(G_2)$, while
their {\em join $G_1 \vee G_2$} is the graph such that $V(G_1 \vee G_2) = V(G_1) \cup V(G_2)$
and $E(G_1 \vee G_2) = E(G_1) \cup E(G_2) \cup \{vw \mid v \in V(G_1), w \in V(G_2) \}$.
Then, a {\em cograph} can be recursively defined as follows:
\begin{listing}{aaa}
	\item[{\rm (1)}] a graph consisting of a single vertex is a cograph;
	\item[{\rm (2)}] if $G_1$ and $G_2$ are cographs, then the union $G_1 \cup G_2$ is a cograph; and
	\item[{\rm (3)}] if $G_1$ and $G_2$ are cographs, then the join $G_1 \vee G_2$ is a cograph.
\end{listing}
	
In this subsection, we show that {\sc dominating set reconfiguration} is solvable in linear time on cographs.
By Theorem~\ref{the:canonical}, it suffices to prove the following lemma.
\begin{lemma} \label{lem:cographcanonical}
Any cograph admits a canonical dominating set.
\end{lemma}

As a proof of Lemma~\ref{lem:cographcanonical}, we will construct a canonical dominating set for any cograph $G$.
By Lemma~\ref{lem:conncted}, it suffices to consider the case where $G$ is
connected and we may assume that $G$ has at least two vertices, because otherwise the problem is trivial.
Then, from the definition of cographs, $G$ must be obtained by the join
operation applied to two cographs $\coga$ and $\cogb$, that is, $G = \coga \vee \cogb$.
Notice that any pair $\{\cographa, \cographb\}$ of vertices $\cographa \in V(\coga)$ and $\cographb \in V(\cogb)$ forms a dominating set of $G$.
Let $\cocanods$ be a dominating set of $G$, defined as follows:

\begin{listing}{a}
	\item[-] If there exists a vertex $w \in V(G)$ such that $N[w] = V(G)$, then let $\cocanods = \{w\}$.
	\item[-] Otherwise choose an arbitrary pair of vertices $\cographa \in V(\coga)$
    and $\cographb \in V(\cogb)$ and let $\cocanods = \{\cographa, \cographb\}$.
\end{listing}

\noindent
Clearly, $\cocanods$ is a minimum dominating set of $G$.
We thus prove the following lemma, which completes the proof of Lemma~\ref{lem:cographcanonical}.

\begin{lemma}[*]\label{lem:cograph03}
For every dominating set $\anyds$ of $G$, $\anyds \sevstep{k} \cocanods$ holds, where $k=|\anyds|+1$.
\end{lemma}

We have thus proved that any cograph has a canonical dominating set.
Then, Theorem~\ref{the:canonical} gives the following corollary.

\begin{corollary}
{\sc DSR} can be solved in linear time on cographs.
\end{corollary}

\subsection{Trees}\label{dsr:altree}

In this subsection, we show that {\sc dominating set reconfiguration} is solvable in linear time on trees.
As for cographs, it suffices to prove the following lemma.

\begin{lemma} \label{lem:treecanonical}
Any tree admits a canonical dominating set.
\end{lemma}

As a proof of Lemma~\ref{lem:treecanonical}, we will construct a canonical dominating set for a tree $T$.
We choose an arbitrary vertex $r$ of degree one in $T$ and regard $T$ as a rooted tree with root $r$.

We first label each vertex in $T$ either $1$, $2$, or $3$, starting from the
leaves of $T$ up to the root $r$ of $T$, as in the following steps (1)--(3);
intuitively, the vertices labeled $2$ will form a dominating set of $T$,
each vertex labeled $1$ will be dominated by its parent, and
each vertex labeled $3$ will be dominated by at least one of its children
(see also \figurename~\ref{fig:tree}(a)):

\begin{listing}{aaa}
	\item[(1)] All leaves in $T$ are labeled $1$.
	\item[(2)] Pick an internal vertex $v$ of $T$, which is not the root, such that all children of $v$ have already been labeled.
					Then,
					\begin{listing}{a}
					\item[-] assign $v$ label $1$ if all children of $v$ are labeled $3$;
					\item[-] assign $v$ label $2$ if at least one child of $v$ is labeled $1$; and
					\item[-] otherwise assign $v$ label $3$.
					\end{listing}
	\item[(3)] Assign the root $r$ (of degree one) label $3$ if its child is labeled $2$, otherwise assign $r$ label $2$.
\end{listing}
For each $i \in \{1,2,3\}$, we denote by $V_i$ the set of all vertices in $T$ that are assigned label $i$.
Then, $\{ \parta, \partb, \partc \}$ forms a partition of $V(T)$.

\begin{figure}[t]
	\centering
		\includegraphics[width=0.85\linewidth]{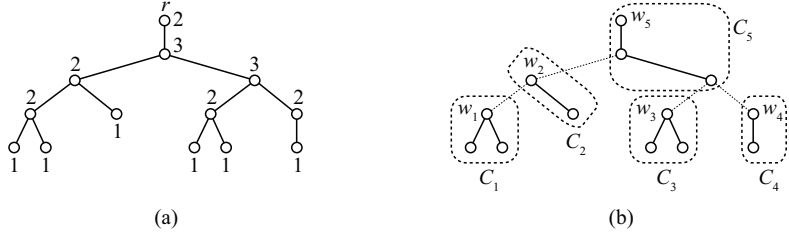}
		\vspace{-1em}
		\caption{(a) The labeling of a tree $T$, and (b) the partition of $V(T)$ into $C_1, C_2, \ldots, C_5$.}
		\vspace{-1em}
	\label{fig:tree}
\end{figure}

We will prove that $\partb$ forms a canonical dominating set of $T$.
We first prove, in Lemmas~\ref{lem:treeds} and \ref{lem:treemin}, that $\partb$ is a minimum dominating set of $T$
and then prove, in Lemma~\ref{lem:treereach}, that $\anyds \sevstep{k} \partb$ holds
for every dominating set $\anyds$ of $T$ and $k = |\anyds|+1$.

\begin{lemma} \label{lem:treeds}
$\partb$ is a dominating set of $T$.
\end{lemma}

\begin{proof}
It suffices to show that both $\parta \subseteq N(\partb)$ and $\partc \subseteq N(\partb)$ hold.

Let $v$ be any vertex in $\parta$, and hence $v$ is labeled $1$.
Then, by the construction above, $v$ is not the root of $T$ and the parent of $v$ must be labeled $2$.
Therefore, $v \in N(\partb)$ holds, as claimed.

Let $u$ be any vertex in $\partc$, and hence $u$ is labeled $3$.
Then, $u$ is not a leaf of $T$. Notice that label $3$ is assigned to a
vertex only when at least one of its children is labeled $2$.
Thus, $u \in N(\partb)$ holds.
%
\qed
\end{proof}

We now prove that $\partb$ is a minimum dominating set of $T$.
To do so, we introduce some notation.
Suppose that the vertices in $\partb$ are ordered as $\vercano{1}, \vercano{2}, \ldots, \vercano{|\partb|}$ by
a post-order depth-first traversal of the tree starting from the root $r$ of $T$.
For each $i \in \{1, 2, \ldots, |\partb| \}$, we denote by $T_i$ the subtree
of $T$ which is induced by $\vercano{i}$ and all its descendants in $T$.
Then, for each $i \in \{1, 2, \ldots, |\partb| \}$, we define a vertex subset $C_i$ of $V(T)$ as follows
(see also \figurename~\ref{fig:tree}(b)):

\begin{eqnarray*}
	C_i = \left\{
	\begin{array}{ll}
		V(T_i) \setminus \bigcup_{j < i} V(T_j) & ~~~\mbox{if $i \neq |\partb|$}; \\
		V(T) \setminus \bigcup_{j < i} V(T_j) & ~~~\mbox{if $i = |\partb|$}.
	\end{array}
	\right.
\end{eqnarray*}

\noindent
Note that $\{ C_1, C_2, \ldots, C_{|\partb|} \}$ forms a partition of $V(T)$.
Furthermore, notice that
\begin{equation} \label{eq:exactlyone}
	\partb \cap C_i = \{ \vercano{i} \}
\end{equation}
holds for every $i \in \{1, 2, \ldots, |\partb| \}$.
Then, Eq.~(\ref{eq:exactlyone}) and the following lemma imply that $\partb$ is a minimum dominating set of $T$.

\begin{lemma}[*]\label{lem:treemin}
Let $\anyds$ be an arbitrary dominating set of $T$.
Then, $|\anyds \cap C_i| \ge 1$ holds for every $i \in \{1, 2, \ldots, |\partb| \}$.
\end{lemma}

We finally prove the following lemma, which completes the proof of Lemma~\ref{lem:treecanonical}.
\begin{lemma}[*]\label{lem:treereach}
For every dominating set $\anyds$ of $T$, $\anyds \sevstep{k} \partb$ holds, where $k = |\anyds| +1$.
\end{lemma}

We have thus proved that $\partb$ forms a canonical dominating set for any tree $T$.
Then, Theorem~\ref{the:canonical} gives the following corollary.
\begin{corollary}
{\sc DSR} can be solved in linear time on trees.
\end{corollary}

\subsection{Interval graphs}\label{dsr:alinterval}
A graph $G$ with $V(G) = \{v_1, v_2, \ldots, v_n\}$ is an {\em interval graph} if
there exists a set $\calI$ of (closed) intervals $I_1, I_2, \ldots, I_n$ such
that $v_i v_j \in E(G)$ if and only if $I_i \cap I_j \neq \emptyset$ for each $i,j \in \{1, 2, \ldots, n \}$.
We call the set $\calI$ of intervals an {\em interval representation} of the graph.
For a given graph $G$, it can be determined in linear time whether $G$ is
an interval graph, and if so obtain an interval representation of $G$~\cite{KM89}.
In this subsection, we show that {\sc dominating set reconfiguration} is solvable in linear time on interval graphs.
As for cographs, it suffices to prove the following lemma.
\begin{lemma}\label{lemma:interval01}
Any interval graph admits a canonical dominating set.
\end{lemma}

As a proof of Lemma~\ref{lemma:interval01}, we will construct a canonical dominating set for any interval graph $G$.
By Lemma~\ref{lem:conncted} it suffices to consider the case where $G$ is connected.
Let $\calI$ be an interval representation of $G$.
For an interval $I \in \calI$, we denote by $\intervala{I}$ and $\intervalb{I}$ the left and right endpoints of $I$, respectively;
we sometimes call the values $\intervala{I}$ and $\intervalb{I}$ the {\em $l$-value} and {\em $r$-value} of $I$, respectively.
%
As for trees, we first label each vertex in $G$ either $1$, $2$, or $3$, from left to right;
the vertices labeled $2$ will form a dominating set of $G$ (see \figurename~\ref{fig:interval} as an example):
\begin{listing}{aaa}
\item[(1)] Pick the unlabeled vertex $v_i$ which has the minimum $r$-value among all unlabeled vertices and assign $v_i$ label $1$.
\item[(2)] Let $v_j$ be the vertex in $N[v_i]$ which has the maximum $r$-value among all vertices in $N[v_i]$.
            Note that $v_j$ may have been already labeled and $v_j = v_i$ may hold.
			We (re)label $v_j$ to $2$.
\item[(3)] For each unlabeled vertex in $N(v_j)$, we assign it label $3$.
\end{listing}

\noindent
We execute steps (1)--(3) above until all vertices are labeled.
For each $i \in \{1,2,3\}$, we denote by $V_i$ the set of all vertices in $G$ that are assigned label $i$.
Then, $\{ \parta, \partb, \partc \}$ forms a partition of $V(G)$.

By the construction above, it is easy to see that $\partb$ forms a dominating set of $G$.
We thus prove that $\partb$ is canonical in Lemmas~\ref{lemma:interval03}
and \ref{lemma:interval04}, that is, $\partb$ is a minimum dominating
set of $G$ (in Lemma~\ref{lemma:interval03}) and $\anyds \sevstep{k} \partb$ holds for
every dominating set $\anyds$ of $G$ and $k=|\anyds|+1$ (in Lemma~\ref{lemma:interval04}).

\begin{figure}[t]
	\centering
		\includegraphics[width=0.85\linewidth]{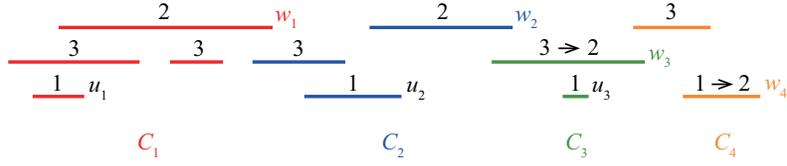}
		\vspace{-1em}
		\caption{The labeling of an interval graph in the interval representation.}
		\vspace{-1em}
	\label{fig:interval}
\end{figure}





We now prove that the dominating set $\partb$ of $G$ is minimum.
To do so, we introduce some notation.
Assume that the vertices in $\partb$ are ordered as $w_1, w_2, \ldots , w_{|\partb|}$ such
that $\intervalb{w_1} < \intervalb{w_2} < \cdots < \intervalb{w_{|\partb|}}$.
For each $i \in \{ 1, 2, \ldots , |\partb| \}$, we define the vertex subset
$C_i$ of $V(G)$ as follows (see \figurename~\ref{fig:interval} as an example):
	
\begin{eqnarray}\label{eq:interval:minimum00}
	C_i = \left\{
			\begin{array}{lrlll}
			\{ v \mid &                            &\intervalb{v} \le \intervalb{w_1} & \} &\mbox{ if $i=1$}; \\
			\{ v \mid & \intervalb{w_{i-1}} < & \intervalb{v} \le \intervalb{w_i} & \} &\mbox{ if $2 \le i \le |\partb|-1$}; \\
			\{ v \mid & \intervalb{w_{|\partb|-1}}  < & \intervalb{v}                     & \} &\mbox{ if $i=|\partb|$}.
			\end{array} \right.
\end{eqnarray}

\noindent
Note that $\{ C_1, C_2, \ldots, C_{|\partb|} \}$ forms a partition of $V(G)$ such that
\begin{eqnarray}\label{eq:interval:minimum01}
	\partb \cap C_i = \{ w_i \}
\end{eqnarray}
holds for every $i \in \{1, 2, \ldots, |\partb| \}$.
Then, Eq.~(\ref{eq:interval:minimum01}) and the following lemma imply that $\partb$ is a minimum dominating set of $G$.
	
\begin{lemma}[*]\label{lemma:interval03}
Let $\anyds$ be an arbitrary dominating set of $G$.
Then, $|\anyds \cap C_i| \ge 1$ holds for every $i \in \{1, 2, \ldots, |\partb| \}$.
\end{lemma}

We finally prove the following lemma, which completes the proof of Lemma~\ref{lemma:interval01}.

\begin{lemma}[*]\label{lemma:interval04}
For every dominating set $\anyds$ of $G$, $\anyds \sevstep{k} \partb$ holds, where $k=|\anyds|+1$.
\end{lemma}

Combining Lemma~\ref{lemma:interval01} and Theorem~\ref{the:canonical} yields the following corollary.
\begin{corollary} \label{cor:interval}
{\sc DSR} can be solved in linear time on interval graphs.
\end{corollary}


\section{Concluding remarks} \label{dsr:conclusion}

In this paper, we delineated the complexity of the
{\sc dominating set reconfiguration} problem restricted to various graph classes.
As shown in \figurename~\ref{fig:results}, our results clarify some
interesting boundaries on the graph classes lying between tractability and PSPACE-completeness:
For example, the structure of interval graphs can be seen as a path-like structure of cliques.
As a super-class of interval graphs, the well-known class of chordal graphs has a tree-like structure of cliques.
We have proved that {\sc dominating set reconfiguration} is solvable in
linear time on interval graphs, while it is PSPACE-complete on chordal graphs.

We note again that our linear-time algorithms for cographs, trees, and interval
graphs employ the same strategy. We also emphasize that this general
scheme can be applied to any graph which admits a canonical dominating set.
It is easy to modify our algorithms so that they actually
find a reconfiguration sequence for a $\YES$-instance $(G, D_s, D_t, k)$ on cographs, trees, or interval graphs.
Observe that each vertex is touched at most once in the
reconfiguration sequence from $D_s$ (or $D_t$) to the canonical dominating set.
Therefore, for a $\YES$-instance on an $n$-vertex graph belonging to one of those classes, there exists a reconfiguration
sequence between $D_s$ and $D_t$ which touches vertices only $O(n)$ times.
In other words, the length of a shortest reconfiguration sequence between $D_s$ and $D_t$ can be bounded by $O(n)$.

\subsubsection*{Acknowledgments.}
This work is partially supported by the Natural Science and
Engineering Research Council of Canada (A.~Mouawad, N.~Nishimura and Y.~Tebbal) and
by MEXT/JSPS KAKENHI 25106504 and 25330003 (T.~Ito), 25104521 and 26540005 (H.~Ono), and 26730001 (A.~Suzuki).

\bibliographystyle{abbrv}

\newpage
\appendix
\section*{Appendix}
\renewcommand{\thesubsection}{\Alph{subsection}}

\subsection{Details omitted from Section~\ref{sec:hardness}}
\subsubsection{Proof of Theorem~\ref{the:split}}
\begin{proof}
We again give a polynomial-time reduction from {\sc vertex cover reconfiguration}.
We extend the idea developed for the NP-hardness proof of {\sc dominating set} on split graphs~\cite{Ber84}.
	
Let $(G^\prime, C_s, C_t, k)$ be an instance of {\sc vertex cover reconfiguration}, where
$V(G^\prime) = \{v_1, v_2, \ldots, v_n\}$ and $E(G^\prime) = \{e_1, e_2, \ldots, e_m\}$.
We construct the corresponding split graph $G$, as follows.
(See also \figurename~\ref{fig:reduction}(a) and (b).)
Let $V(G) = A \cup B$, where $A = V(G^\prime)$ and $B = \{ w_1, w_2, \ldots, w_m\}$;
each vertex $w_i \in B$ corresponds to the edge $e_i$ in $E(G^\prime)$.
We join all pairs of vertices in $A$ so that $A$ forms a clique in $G$.
In addition, for each edge $e_i = v_p v_q$ in $E(G^\prime)$, we join $w_i \in B$ with each of $v_p$ and $v_q$ in $G$.
Let $G$ be the resulting graph, and let $(G, D_s = C_s, D_t = C_t, k)$ be the
corresponding instance of {\sc dominating set reconfiguration}.
Clearly, this instance can be constructed in polynomial time.
Thus, we will prove that $D_s \sevstep{k} D_t$ holds if and only if there is a
reconfiguration sequence of vertex covers in $G^\prime$ between $C_s$ and $C_t$.

We first prove the if direction.
Because both problems employ the same reconfiguration rule, it suffices to prove
that any vertex cover $C$ of $G^\prime$ forms a dominating set of $G$.
Since $C \subseteq V(G^\prime) = A$ and $A$ is a clique, all vertices in $A$ are dominated by the vertices in $C$.
Thus, consider a vertex $w_i$ in $B$, which corresponds to the edge $e_i = v_p v_q$ in $E(G^\prime)$.
Then, since $C$ is a vertex cover of $G^\prime$, at least one of $v_p$ and $v_q$ must be contained in $C$.
This means that $w_i$ is dominated by the endpoint $v_p$ or $v_q$ in $G$.
Therefore, $C$ is a dominating set of $G$.
	
\begin{figure}
    \centering
	\includegraphics[width=0.9\linewidth]{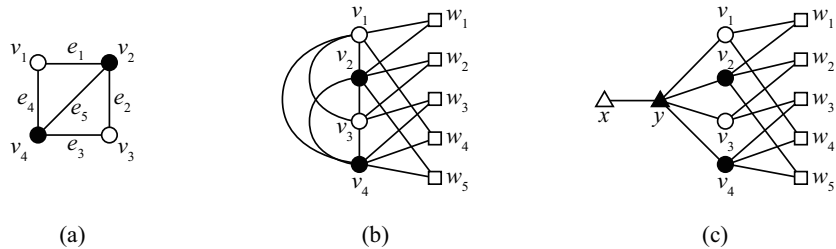}
	\vspace{-1em}
	\caption{(a) Vertex cover $\{v_2, v_4\}$ of a graph, (b) dominating set
    $\{v_2, v_4\}$ of the corresponding split graph, and (c) dominating set $\{v_2, v_4, y\}$ of the corresponding bipartite graph.}
	\vspace{-1em}
	\label{fig:reduction}
\end{figure}

We now prove the only-if direction.
Notice that, for each vertex $w_i \in B$ corresponding to
the edge $e_i = v_p v_q$ in $E(G^\prime)$, we have $N_G[w_i] \subseteq N_G[v_p]$ and $N_G[w_i] \subseteq N_G[v_q]$.
Therefore, if $D_s \sevstep{k} D_t$ holds, then we can obtain a reconfiguration
sequence of dominating sets in $G$ between $D_s$ and $D_t$ which touches vertices only in $A = V(G^\prime)$;
recall the arguments in the proof of Theorem~\ref{the:hardness1}.
Observe that any dominating set $D$ of $G$ such that $D \subseteq A = V(G^\prime)$
forms a vertex cover of $G^\prime$, because each vertex $w_i \in B$ is dominated by at least one vertex in $C \subseteq V(G^\prime)$.
We have thus verified the only-if direction.
\qed
\end{proof}

\subsubsection{Proof of Theorem~\ref{the:bipartite}}
\begin{proof}
We give a polynomial-time reduction from {\sc dominating set reconfiguration}
on split graphs to the same problem restricted to bipartite graphs.
The same idea is used in the NP-hardness proof of {\sc dominating set} for bipartite graphs~\cite{Ber84}.
	
Let $(G^\prime, D_s^\prime, D_t^\prime, k^\prime)$ be an instance
of {\sc dominating set reconfiguration}, where $G^\prime$ is a split graph.
Then, $V(G^\prime)$ can be partitioned into two subsets $A$ and $B$ which
form a clique and an independent set in $G^\prime$, respectively.
Furthermore, by the reduction given in the proof of Theorem~\ref{the:split}, the problem
for split graphs remains PSPACE-complete even if both $D_s^\prime \subseteq A$ and $D_t^\prime \subseteq A$ hold. 	
	
We now construct the corresponding bipartite graph $G$, as follows.
(See also \figurename~\ref{fig:reduction}(b) and (c).)
First, we delete any edge joining two vertices in $A$, and make $A$ an independent set.
Then, we add a new edge consisting of two new vertices $x$ and $y$ and join $y$ with each vertex in $A$.
The resulting graph $G$ is bipartite.
Let $D_s = D_s^\prime \cup \{y\}$, $D_t = D_t^\prime \cup \{y\}$, $k = k^\prime+1$, and we
obtain the corresponding {\sc dominating set reconfiguration} instance $(G, D_s, D_t, k)$, where $G$ is bipartite.
Clearly, this instance can be constructed in polynomial time.
Thus, we will prove that $D_s \sevstep{k} D_t$ holds if and only if $D_s^\prime \sevstep{k^\prime} D_t^\prime$ holds.

We first prove the if direction.
Suppose that $D_s^\prime \sevstep{k^\prime} D_t^\prime$ holds. Hence, there exists
a reconfiguration sequence in $G^\prime$ between $D_s^\prime$ and $D_t^\prime$.
Consider any dominating set $D^\prime$ of $G^\prime$ in this sequence.
Then, $B \subset N_{G}[D^\prime]$ holds because $B \subset N_{G^\prime}[D^\prime]$ and we
have deleted only the edges such that both endpoints are in $A$.
Since $N_G(y) = A \cup \{x\}$, we can conclude that $D^\prime \cup \{y\}$ is a dominating set of $G$.
Furthermore, $|D^\prime \cup \{y\}| \le k^\prime + 1 = k$.
Thus, $D_s \sevstep{k} D_t$ holds.

We then prove the only-if direction.
Suppose that $D_s \sevstep{k} D_t$ holds, and hence there exists a
reconfiguration sequence in $G$ between $D_s = D_s^\prime \cup \{y\}$ and $D_t = D_t^\prime \cup \{y\}$.
Notice that any dominating set of $G$ contains at least one of $x$ and $y$.
Since $N_G[x] \subset N_G[y]$ and $y \in D_s, D_t$, we can assume that $y$ is
contained in all dominating sets in the reconfiguration sequence.
Recall that both $D_s^\prime \subseteq A$ and $D_t^\prime \subseteq A$ hold.
Thus, if a vertex $w_i \in B$ is touched, then it must be added first.
Since $N_G(y) = A \cup \{x\}$, we have $N_G[\{w_i, y\}] = N_G[\{v_p, y\}] = N_G[\{v_q, y\}]$, where $N_G(w_i) = \{ v_p, v_q\}$.
Therefore, we can replace the addition of $w_i$ by that of either $v_p$ or $v_q$ and
obtain a reconfiguration sequence in $G$ between $D_s$ and $D_t$ which touches vertices only in $A$.
Consider any dominating set $D$ of $G$ in such a reconfiguration sequence.
Since $y \in D$, we have $|D \cap V(G^\prime)| \le k - 1 = k^\prime$.
Furthermore, since $D \cap V(G^\prime) \subseteq A$ and $A$ forms a clique
in $G^\prime$, we have $A \subseteq N_{G^\prime}[D \cap V(G^\prime)]$.
Since there is no edge joining $y$ and a vertex in $B$, each vertex in
$B$ is dominated by some vertex in $D \cap V(G^\prime)$.
Therefore, $D \cap V(G^\prime)$ is a dominating set of $G^\prime$ of
cardinality at most $k^\prime$ and $D_s^\prime \sevstep{k^\prime} D_t^\prime$ holds.
\qed
\end{proof}

\subsection{Details omitted from Section~\ref{sec:algo}}
\subsubsection{Proof of Lemma~\ref{lem:conncted}}
\begin{proof}
Let $\anyds$ be any dominating set of $G$.
For each $i \in \{1,2,\ldots,\cocompnum\}$, since $\cocanods_i$ is
canonical for $\cocompg{i}$, we have $\anyds \cap V(\cocompg{i}) \sevstep{k_i} \cocanods_i$ for $k_i = |\anyds \cap V(\cocompg{i})|+1$.
 Therefore, we can independently transform $\anyds \cap V(\cocompg{i})$ into $\cocanods_i$ for each $i \in \{1,2,\ldots,\cocompnum\}$.
Clearly, this is a reconfiguration sequence from $\anyds$ to $\cocanods = \cocanods_1 \cup \cocanods_2 \cup \cdots \cup \cocanods_{\cocompnum}$.
Furthermore, since $\cocanods_i$ is a minimum dominating set of $\cocompg{i}$, we have
$|\anyds \cap V(\cocompg{i})| \ge |\cocanods_i|$ for each $i \in \{1,2,\ldots,\cocompnum\}$.
Thus, any dominating set appearing in the sequence is of cardinality at most $|\anyds|+1$.
\qed
\end{proof}

\subsection{Details omitted from Section~\ref{dsr:cograph}}
\subsubsection{Proof of Lemma~\ref{lem:cograph03}}
\begin{proof}
We construct a reconfiguration sequence from $\anyds$ to $\cocanods$ such that
each intermediate dominating set is of cardinality at most $|\anyds|+1$.
\medskip

\noindent
{\bf Case (i):} $|\cocanods| = 1$.
	
In this case, $\cocanods$ consists of a universal vertex $w$, that is, $N[w] = V(G)$.
Therefore, we first add $w$ to $\anyds$ if $w \not\in \anyds$, and then delete the vertices in $\anyds \setminus \{w\}$ one by one.
Since $N[w] = V(G)$, all intermediate vertex subsets are dominating sets of $G$.
Since the addition is applied only to $w$, we have $\anyds \sevstep{k} \cocanods$ for $k=|\anyds|+1$.
\medskip

\noindent
{\bf Case (ii):} $|\cocanods| = 2$.

In this case, $\cocanods$ consists of two vertices $\cographa \in V(\coga)$ and $\cographb \in V(\cogb)$.
Since $\cocanods$ is a minimum dominating set of $G$, we have $|\anyds| \ge 2$.
Note that, however, $\anyds \subseteq V(\coga)$ or $\anyds \subseteq V(\cogb)$ may hold.
We assume without loss of generality that $|\anyds \cap V(\coga)| \ge |\anyds \cap V(\cogb)|$.
Then, we construct a sequence of vertex subsets of $G$, as follows:
\begin{listing}{aaa}
	\item[(1)] Add $\cographb$ to $\anyds$ if $\cographb \not\in \anyds$; let $\anyds_1 = \anyds \cup \{\cographb\}$.
	\item[(2)] If $|\anyds_1 \cap V(\coga)| = |\anyds \cap V(\coga)| \ge 2$, then delete
    one vertex in $\anyds \cap (V(\coga) \setminus \{\cographa\})$; otherwise delete a vertex
    in $\anyds_1 \cap (V(\cogb) \setminus \{\cographb\}) = \anyds \cap (V(\cogb) \setminus \{\cographb\})$ if it exists.
	Let $\anyds_2$ be the resulting vertex subset of $G$.
	\item[(3)] Add $\cographa$ to $\anyds_2$ if $\cographa \not\in \anyds_2$; let $\anyds_3 = \anyds_2 \cup \{ \cographa\}$.
	\item[(4)] Delete from $\anyds_3$ all vertices in $\anyds \setminus \{\cographa, \cographb\}$ one by one.
\end{listing}
We will prove that each vertex subset appearing above is a dominating set of $G$ with cardinality at most $|\anyds|+1$.
Indeed, it suffices to show that $\anyds_2$ is a dominating set of $G$ such that $|\anyds_2| \le |\anyds|$;
note that $\anyds_3$ contains both $\cographa \in V(\coga)$ and $\cographb \in V(\cogb)$ and hence
any vertex subset appearing in Steps~(3) and (4) above is a dominating set of $G$ with cardinality at most $|\anyds_2|+1$.

We first consider the case where $|\anyds_1 \cap V(\coga)| \ge 2$.
In this case, $\anyds_1 \cap (V(\coga) \setminus \{\cographa\}) \neq \emptyset$, and hence we
can delete one vertex $u$ $(\neq \cographa)$ from $\anyds_1$.
We thus have $|\anyds_2| = |\anyds_1| - 1 \le |\anyds|$, as required.
Since $|\anyds_1 \cap V(\coga)| \ge 2$, $\anyds_2$ $(= \anyds_1 \setminus \{u\})$ contains at least one vertex in $V(\coga)$.
Furthermore, $\cographb \in \anyds_2$ and hence $\anyds_2$ is a dominating set of $G$.
	
We then consider the case where $|\anyds_1 \cap V(\coga)| \le 1$.
Note that, since $|\anyds| \ge 2$ and $|\anyds_1 \cap V(\coga)| = |\anyds \cap V(\coga)| \ge |\anyds \cap V(\cogb)|$, we
have $|\anyds \cap V(\coga)| = |\anyds \cap V(\cogb)| = 1$ in this case.
Let $\anyds \cap V(\cogb) = \{ z\}$.
If $\cographb \not\in \anyds$ (and hence $z \neq \cographb$) then $|\anyds_1| = |\anyds|+1$
and $\anyds_1 \cap (V(\cogb) \setminus \{\cographb\}) = \{z\}$.
Therefore, $\anyds_2 = \anyds_1 \setminus \{z\}$ and $|\anyds_2| = |\anyds_1| - 1 = |\anyds|$.
Furthermore, since $\cographb \in \anyds_2$ and $|\anyds_2 \cap V(\coga)| = |\anyds_1 \cap V(\coga)| = 1$, $\anyds_2$ is a dominating set of $G$.
On the other hand, if $\cographb \in \anyds$, then we have $\anyds \cap V(\cogb) = \{ \cographb \}$.
Consequently, $\anyds_2 = \anyds_1 = \anyds$ and hence $\anyds_2$ is a dominating set of $G$ of cardinality $|\anyds_2| = |\anyds|$.
\qed
\end{proof}

\subsection{Details omitted from Section~\ref{dsr:altree}}
\subsubsection{Proof of Lemma~\ref{lem:treemin}}
\begin{proof}
Suppose for a contradiction that $\anyds \cap C_i = \emptyset$ holds for some index $i \in \{1, 2, \ldots, |\partb| \}$.
We will prove that $C_i$ contains at least one vertex $u$ such that $N[u] \subseteq C_i$.
Then, since $\anyds \cap C_i = \emptyset$, the vertex $u$ is not dominated by any vertex in $D$;
this contradicts the assumption that $\anyds$ is a dominating set of $T$.
Recall that all leaves in $T$ are labeled $1$, and hence $\vercano{i}$ is an internal vertex.
	
First, consider the case where $\vercano{i}$ has a child $u$ which is a leaf of $T$.
Then, $N[u] \subseteq C_i$ holds for the leaf $u$; a contradiction.

Second, consider the case where $i = |\partb|$, that is, $C_i$ $\bigl(=C_{|\partb|} \bigr)$ contains the root $r$ of $T$.
Recall that $r$ is of degree one and is labeled either $2$ or $3$;
we will prove that $N[r] \subseteq C_{|\partb|}$ holds.
If $r$ is labeled $3$, then its (unique) child $v$ is labeled $2$ and hence $v = \vercano{|\partb|}$.
Therefore, $C_{|\partb|}$ contains both $r$ and $v$ and hence $N[r] \subseteq C_{|\partb|}$ holds; a contradiction.
On the other hand, if $r$ is labeled $2$ and hence $r = \vercano{|\partb|}$, then its child $v$ is labeled either $1$ or $3$.
Therefore, $C_{|\partb|}$ contains both $r$ and $v$, and hence $N[r] \subseteq C_{|\partb|}$ holds; a contradiction.

Finally, consider the case where $i \neq |\partb|$ and $\vercano{i}$ is an
internal vertex such that all children of $\vercano{i}$ are also internal vertices in $T$.
Since $\vercano{i}$ is labeled $2$, there exists at least one child $u$ of $\vercano{i}$ which is labeled $1$.
Then, since $u$ is an internal vertex, all children of $u$ (and hence all ``grandchildren'' of $\vercano{i}$) are labeled $3$.
Therefore, $N[u] \subseteq C_i$ holds for the child $u$ of $\vercano{i}$; a contradiction.
\qed
\end{proof}

\subsubsection{Proof of Lemma~\ref{lem:treereach}}
\begin{proof}
We construct a reconfiguration sequence from $\anyds$ to $\partb$
such that each intermediate dominating set is of cardinality at most $|\anyds|+1$.
	
Let $\anyds_0 = \anyds$.
For each $i$ from $1$ to $|\partb|$, we focus on the vertices in $C_i$
and transform $\anyds_{i-1} \cap C_i$ into $\partb \cap C_i$ as follows:

\begin{listing}{aaa}
	\item[(1)] add the vertex $\vercano{i} \in \partb \cap C_i$ to $\anyds_{i-1}$ if $\vercano{i} \notin \anyds_{i-1}$;
	\item[(2)] delete the vertices in $\anyds_{i-1} \cap \bigl(C_i \setminus \{\vercano{i}\} \bigr)$ one by one; and
	\item[(3)] let $\anyds_i$ be the resulting vertex set.
\end{listing}
\smallskip
	
We first claim that $\anyds_i$ forms a dominating set of $T$ for each $i \in \{1, 2, \ldots, |\partb| \}$.
Notice that $\anyds_i \cap V(T_i) = \partb \cap V(T_i)$ for the resulting
vertex set $\anyds_i$. Moreover, only the root $\vercano{i}$ of $T_i$ is adjacent to a vertex in $V(T) \setminus V(T_i)$.
Since $\vercano{i} \in \partb$ and both $\partb$ and $\anyds_{i-1}$ form
dominating sets of $T$, we can conclude that $\anyds_i$ forms a dominating set of $T$.
Then, all vertex subsets appearing in Steps~(1) and (2) above also form
dominating sets of $T$, because each of them is a superset of $\anyds_i$.

We then claim that $|\anyds_{i-1}| \ge |\anyds_i|$ for each $i \in \{1, 2, \ldots, |\partb| \}$.
If $\vercano{i} \in \anyds_{i-1}$, then the claim clearly holds because
we only delete vertices in Step~(2) without adding the vertex $\vercano{i}$ in Step~(1).
We thus consider the case where $\vercano{i} \not\in \anyds_{i-1}$.
Since $\anyds_{i-1}$ is a dominating set of $T$, Lemma~\ref{lem:treemin} implies
that $\anyds_{i-1} \cap \bigl(C_i \setminus \{\vercano{i}\} \bigr) \neq \emptyset$ in this case.
Therefore, we have $|\anyds_{i-1}| \ge |\anyds_i|$.

Note that, since addition is executed only in Step~(1), the maximum cardinality of any dominating set
in the reconfiguration sequence from $\anyds_{i-1}$ to $\anyds_i$ is at most $|\anyds_{i-1}| + 1$.
Since $|\anyds_{i-1}| \ge |\anyds_i|$ for each $i \in \{1, 2, \ldots, |\partb| \}$, the maximum cardinality
of any dominating set in the reconfiguration sequence
from $\anyds_{0}$ $(=\anyds)$ to $\anyds_{|\partb|}$ $(=\partb)$ is at most $|\anyds| + 1$.
Therefore, there exists a reconfiguration sequence from $\anyds$ to $\partb$ such that all
intermediate dominating sets are of cardinality at most $|\anyds|+1$.
\qed
\end{proof}

\subsection{Details omitted from Section~\ref{dsr:alinterval}}
\subsubsection{Proof of Lemma~\ref{lemma:interval03}}
\begin{proof}
Suppose for a contradiction that $\anyds \cap C_i = \emptyset$ holds for some index $i \in \{1, 2, \ldots, |\partb| \}$.
Assume that the vertices in $\parta$ are ordered as $u_1, u_2, \ldots, u_{|\parta|}$
such that $\intervalb{u_1} < \intervalb{u_2} < \cdots < \intervalb{u_{|\parta|}}$.
Then, observe that $\parta \cap C_i = \{ u_i \}$ holds for every $i \in \{1, 2, \ldots, |\parta| \}$.
In addition, $\parta \cap C_{|\partb|} = \emptyset$ holds if $|\partb|=|\parta|+1$.

First, we consider the case where both $i = |\partb|$ and $|\partb| = |\parta|+1$ hold;
in this case, both $\parta \cap C_{|\partb|} = \emptyset$ and $\partb \cap C_{|\partb|} = \{ w_{|\partb|} \}$ hold.
Since $\anyds \cap C_{|\partb|} = \emptyset$, $w_{|\partb|} \in \partb$ must be
dominated by some vertex $v$ in $C_{-} = C_1 \cup C_2 \cup \cdots \cup C_{|\partb|-1}$.
Then, $vw_{|\partb|} \in E(G)$ and hence we have $ \intervala{w_{|\partb|}} \le \intervalb{v}$.
Since $v \in C_-$, by Eq.~(\ref{eq:interval:minimum00}) we have $\intervalb{v} \le \intervalb{w_{|\partb|-1}}$
and hence $\intervala{w_{|\partb|}} \le \intervalb{w_{|\partb|-1}} < \intervalb{w_{|\partb|}}$.
Therefore, $w_{|\partb|} \in N(w_{|\partb|-1})$ holds and $w_{|\partb|}$ must be labeled $3$.
This contradicts the assumption that $w_{|\partb|}$ is labeled $2$.

We now consider the other case, that is, both $\parta \cap C_i = \{ u_i \}$
and $\partb \cap C_i = \{ w_i \}$ hold for index $i$.
Since $\anyds \cap C_i = \emptyset$, $u_i \in \parta$ must be dominated
by at least one vertex in $C_- = C_1 \cup C_2 \cup \cdots \cup C_{i-1}$ or $C_+ = C_{i+1} \cup C_{i+2} \cup \cdots \cup C_{|\partb|}$.
If $u_i$ is dominated by some vertex in $C_-$, then the same arguments given
above yield a contradiction, i.e. $u_i$ must be labeled $3$ even though $u_i$ is in $\parta$.
Therefore, $u_i$ must be dominated by some vertex $v$ in $C_+$.
Then, since $vu_i \in E(G)$, we have $v \in N(u_i) \subset N[u_i]$.
Furthermore, since $v \in C_+$, by Eq.~(\ref{eq:interval:minimum00}) we have $ \intervalb{w_i} < \intervalb{v}$.
However, recall that $w_i \in \partb$ is chosen as the vertex in $N[u_i]$ which has the maximum $r$-value among all vertices in $N[u_i]$.
This contradicts the assumption that $w_i$ is labeled $2$.
\qed
\end{proof}

\subsubsection{Proof of Lemma~\ref{lemma:interval04}}
\begin{proof}
We construct a reconfiguration sequence from $\anyds$ to $\partb$ such that
each intermediate dominating set is of cardinality at most $|\anyds|+1$.

Let $\anyds_0 = \anyds$.
For each $i$ from $1$ to $|\partb|$, we focus on the vertices in $C_i$, and transform $\anyds_{i-1} \cap C_i$ into $\partb \cap C_i$ as follows:
\begin{listing}{aaa}
	\item[(1)] add the vertex $w_i \in \partb \cap C_i$ to $\anyds_{i-1}$ if $w_i \notin \anyds_{i-1}$;
	\item[(2)] delete the vertices in $\anyds_{i-1} \cap (C_i \setminus \{ w_i \})$ one by one; and
	\item[(3)] let $\anyds_i$ be the resulting vertex set.
\end{listing}
\smallskip
	
For each $i \in \{1, 2, \ldots, |\partb| \}$, let $C_{-i} = C_1 \cup C_2 \cup \cdots \cup C_i$
and $C_+ = C_{i+1} \cup C_{i+2} \cup \cdots \cup C_{|\partb|}$.
We claim that $\anyds_i$ forms a dominating set of $G$:
\begin{listing}{a}
	\item[-] Consider a vertex $v$ such that $\intervalb{v} \le \intervalb{w_i}$.
				Since $\anyds_i \cap C_{-i} = \partb \cap C_{-i}$ holds, $v$ is dominated by some vertex in $\partb \cap C_{-i}$.
	\item[-] Consider a vertex $v$ such that $\intervalb{w_i} \le \intervala{v}$.
				Since $\anyds_i \cap C_+ = \anyds \cap C_+$ holds, $v$ is dominated by some vertex in $\anyds \cap C_+$.
	\item[-] Finally, consider a vertex $v$ such that $\intervala{v} < \intervalb{w_i} < \intervalb{v}$.
				Then, $vw_i \in E(G)$ and hence $v$ is dominated by $w_i \in \anyds_i$.
\end{listing}
Thus, $\anyds_i$ forms a dominating set of $G$.
Since each vertex subset appearing in Steps~(1) and (2) above is a superset of $\anyds_i$, it also forms a dominating set of $G$.

By the same arguments as in the proof of Lemma~\ref{lem:treereach}, we can conclude
that the reconfiguration sequence from $\anyds$ to $\partb$ above consists
only of dominating sets of cardinality at most $|\anyds|+1$.
\qed
\end{proof}

\end{document}